\documentclass[conference]{IEEEtran}
\usepackage[ansinew]{inputenc}
\usepackage[pdftex]{graphicx}
\usepackage[T1]{fontenc}
\usepackage{color,amsmath,enumerate,amssymb,hhline,verbatim}

\begin{document}

\title{Exact-Regenerating Codes between MBR and MSR Points}

\author{\IEEEauthorblockN{Toni Ernvall}
\IEEEauthorblockA{Turku Center for Computer Science\\ \& Department of Mathematics and Statistics\\ FI-20014 University of Turku\\ Finland\\
Email: {tmernv}@utu.fi.}
}
\maketitle

\newtheorem{definition}{Definition}[section]
\newtheorem{thm}{Theorem}[section]
\newtheorem{proposition}[thm]{Proposition}
\newtheorem{lemma}[thm]{Lemma}
\newtheorem{corollary}[thm]{Corollary}
\newtheorem{exam}{Example}[section]
\newtheorem{conj}{Conjecture}
\newtheorem{remark}{Remark}[section]

\newcommand{\La}{\mathbf{L}}
\newcommand{\h}{{\mathbf h}}
\newcommand{\Z}{{\mathbf Z}}
\newcommand{\R}{{\mathbf R}}
\newcommand{\C}{{\mathbf C}}
\newcommand{\D}{{\mathcal D}}
\newcommand{\F}{{\mathbf F}}
\newcommand{\HH}{{\mathbf H}}
\newcommand{\OO}{{\mathcal O}}
\newcommand{\G}{{\mathcal G}}
\newcommand{\A}{{\mathcal A}}
\newcommand{\B}{{\mathcal B}}
\newcommand{\I}{{\mathcal I}}
\newcommand{\E}{{\mathcal E}}
\newcommand{\PP}{{\mathcal P}}
\newcommand{\Q}{{\mathbf Q}}
\newcommand{\M}{{\mathcal M}}
\newcommand{\separ}{\,\vert\,}
\newcommand{\abs}[1]{\vert #1 \vert}

\begin{abstract}
In this paper we study distributed storage systems with exact repair. We give a construction for regenerating codes between the minimum storage regenerating (MSR) and the minimum bandwidth regenerating (MBR) points and show that in the case that the parameters $n$, $k$, and $d$ are close to each other our constructions are close to optimal when comparing to the known capacity when only functional repair is required. We do this by showing that when the distances of the parameters $n$, $k$, and $d$ are fixed but the actual values approach to infinity, the fraction of the performance of our codes with exact repair and the known capacity of codes with functional repair approaches to one.
\end{abstract}

\section{Introduction}
\subsection{Regenerating Codes}
In a distributed storage system a file is dispersed across $n$ nodes in a network such that given any $k\, (<n)$ of these nodes one can reconstruct the original file. We also want to have such a redundancy in our network that if we lose a node then any $d\, (<n)$ of the remaining nodes can repair the lost node. We assume that each node stores the amount $\alpha$  of information, \emph{e.g.}, $\alpha$ symbols over a finite field, and in the repair process each repairing node transmits the amount $\beta$ to the new replacing node (called a \emph{newcomer}) and hence the total repair bandwidth is $\gamma=d \beta$. We also assume that $k \leq d$.

The repair process can be either functional or exact. By functional repair we mean that the nodes may change over time, \emph{i.e.}, if a node $v_{i}^{\text{old}}$ is lost and in the repair process we get a new node $v_{i}^{\text{new}}$ instead, then we may have $v_{i}^{\text{old}} \neq v_{i}^{\text{new}}$. If only functional repair is assumed then the capacity of the system, denoted by $C_{k,d}(\alpha,\gamma)$, is known. Namely, it was proved in the pioneering work by Dimakis \emph{et al.} \cite{dimakis} that
$$
C_{k,d}(\alpha,\gamma)=\sum_{j=0}^{k-1} \min \left\{ \alpha, \frac{d-j}{d}\gamma  \right\}.
$$

If the size of the stored file is fixed to be $B$ then the above expression for the capacity defines a tradeoff between the node size $\alpha$ and the total repair bandwidth $\gamma$. The two extreme points are called the minimum storage regeneration (MSR) point and the minimum bandwidth regeneration (MBR) point. The MSR point is achieved by first minimizing $\alpha$ and then minimizing $\gamma$ to obtain
\begin{equation}\label{MSR}
\left\{
  \begin{array}{l}
\alpha = \frac{B}{k} \\
\gamma = \frac{d B}{k(d-k+1)}.
  \end{array} \right.
\end{equation}
By first minimizing $\gamma$ and then minimizing $\alpha$ leads to the MBR point
\begin{equation}\label{MBR}
\left\{
  \begin{array}{l}
\alpha = \frac{2d B}{k(2d-k+1)} \\
\gamma = \frac{2d B}{k(2d-k+1)}.
  \end{array} \right.
\end{equation}

In this paper we are interested in codes that have exact repair. The concepts of exact regeneration and exact repair were introduced independently in \cite{explicitconst}, \cite{reducingrepair}, and \cite{searchingfor}. Exact repair means that the network of nodes does not vary over time, \emph{i.e.}, if a node $v_{i}^{\text{old}}$ is lost and in the repair process we get a new node $v_{i}^{\text{new}}$, then $v_{i}^{\text{old}} = v_{i}^{\text{new}}$. We denote by
$$
C_{n,k,d}^{\text{exact}} (\alpha,\gamma)
$$
 the capacity of codes with exact repair with $n$ nodes each of size $\alpha$, with total repair bandwidth $\gamma$, and for which each set of $k$ nodes can recover the stored file and each set of $d$ nodes can repair a lost node.

We have by definition that
$$
C_{n,k,d}^{\text{exact}} (\alpha,\gamma) \leq C_{k,d} (\alpha,\gamma).
$$
It was proved in \cite{kumar}, \cite{MSRequal}, and \cite{optimalMDS} that the codes with exact repair achieve the MSR point and in \cite{kumar} that the codes with exact repair achieve the MBR point. The impossibility of constructing codes with exact repair at essentially all interior points on the storage-bandwidth tradeoff curve was shown in \cite{nonachievability}.

\subsection{Contributions and Organization}
In Section \ref{construction} we give a construction for codes between MSR and MBR points with exact repair. In Section \ref{inequalities} we derive some inequalities from our construction. Section \ref{example} provides an example showing  that, in the special case of $n=k+1=d+1$, our construction is close to optimal when comparing to the known capacity when only functional repair is required. In Section \ref{analysis} we show that when the distances of the parameters $n$, $k$, and $d$ are fixed but the actual values approach to infinity, the fraction of performance of our codes with exact repair and the known capacity of functional-repair codes approaches to one.

\section{Construction}\label{construction}
Assume we have a storage system $DSS_1$ with exact repair for parameters $$(n,k,d)$$ with a node size $\alpha$ and the total repair bandwidth $\gamma=d\beta$. In this section we propose a construction that gives a new storage system for parameters $$(n'=n+1,k'=k+1,d'=d+1).$$ Let $DSS_1$ consist of nodes $v_1,\dots,v_n$, and let the stored file $F$ be of maximal size $C^{\text{exact}}_{n,k,d}(\alpha,\gamma)$.

Let then $DSS_{1+}$ denote a new system consisting of the original storage system $DSS_1$ and one extra node $v_{n+1}$ storing nothing. It is clear that $DSS_{1+}$ is a storage system for parameters $$(n+1,k+1,d+1)$$ and can store the original file $F$.

Let $\sigma_j$ be the permutations of the set $\{ 1, \dots, n+1 \}$ for $j=1,\dots, (n+1)!$\,. Assume that $DSS_{j}^{\text{new}}$ is a storage system for $j=1,\dots,(n+1)!$ corresponding to the permutation $\sigma_j$ such that $DSS_{j}^{\text{new}}$ is  exactly the same as $DSS_{1+}$ except that the order of the nodes is changed corresponding to the permutation $\sigma_j$, \emph{i.e.}, the $i$th node in $DSS_{1+}$ is the $\sigma_j(i)$th node in $DSS_{j}^{\text{new}}$.

Using these $(n+1)!$ new systems as building blocks we construct a new system $DSS_2$ such that its $j$th node for $j=1,\dots,n+1$ stores the $j$th node from each system $DSS_{i}^{\text{new}}$ for $i=1,\dots,(n+1)!$\,.

It is clear that this new system $DSS_2$ works for parameters $(n+1,k+1,d+1)$, has exact repair property, stores a file of size $(n+1)! C^{\text{exact}}_{n,k,d}(\alpha,\gamma)$ and has a node size
$$
\alpha_2=((n+1)!-n!)\alpha=n\cdot n!\alpha
$$
and total repair bandwidth
$$
\gamma_2=((n+1)!-n!)\gamma=n\cdot n!\gamma\,.
$$
Moreover, because of the symmetry of the construction we have $\beta_2=n\cdot n!\beta$\,.

This construction implies the inequality
$$
C^{\text{exact}}_{n+1,k+1,d+1}(n\cdot n!\alpha,n\cdot n!\gamma) \geq (n+1)! C^{\text{exact}}_{n,k,d}(\alpha,\gamma),
$$
that is,
\begin{equation}\label{firstbound}
C^{\text{exact}}_{n+1,k+1,d+1}(\alpha,\gamma) \geq \frac{n+1}{n} C^{\text{exact}}_{n,k,d}(\alpha,\gamma).
\end{equation}

\begin{exam}\label{helppoesim}
If we relax on the typical requirement of a DSS to be homogeneous, meaning that each node is transmitting the same amount $\beta$ of information in the repair process, and instead only require that the total repair bandwidth $\gamma$ is constant (\emph{i.e.,} $\beta$ may take different values depending on the node), then we can build our construction a little easier. Let $(n,k,d)=(3,2,2)$ and $DSS_1$ be a distributed storage system with exact repair. Let $DSS_{j}^{\text{new}}$ be a storage system with $4$ nodes for $j=1,\dots,4$ where the $j$th node stores nothing, the $i$th node for $i<j$ stores as the $i$th node in the original system $DSS_1$, and the $i$th node for $i>j$ stores as the $(i-1)$th node in the original system $DSS_1$. That is, in the $j$th subsystem $DSS_{j}^{\text{new}}$ the $j$th node stores nothing while the other nodes are as those in the original system $DSS_1$.

Using these four new systems as building blocks we construct a new system $DSS_2$ such that its $j$th node for $j=1,\dots,4$ stores the $j$th node from each system $DSS_{i}^{\text{new}}$ for $i=1,\dots,4$. Hence each node in $DSS_2$ stores $(4-1)\alpha=3\alpha$ and the total repair bandwidth is $(4-1)\gamma=3\gamma$.

For example, if the original system $DSS_1$ consists of nodes $v_1$ storing $x$, $v_2$ storing $y$, and $v_3$ storing $x+y$ then $DSS_{1}^{\text{new}}$ consists of nodes $u_{11}$ storing nothing, $u_{12}$ storing $x_1$, $u_{13}$ storing $y_1$, and $u_{14}$ storing $x_1+y_1$. Similarly $DSS_{2}^{\text{new}}$ consists of nodes $u_{21}$ storing $x_2$, $u_{22}$ storing nothing, $u_{23}$ storing $y_2$, and $u_{24}$ storing $x_2+y_2$ and so on. Then in the resulting system the first node $w_1$ consists of nodes $u_{11}$ (storing nothing), $u_{21}$ (storing $x_2$), $u_{31}$ (storing $x_3$), and $u_{41}$ (storing $x_4$). The stored file is $(x_1,x_2,x_3,x_4,y_1,y_2,y_3,y_4)$.
\end{exam}

\begin{figure}[ht]
     \includegraphics[width=9cm]{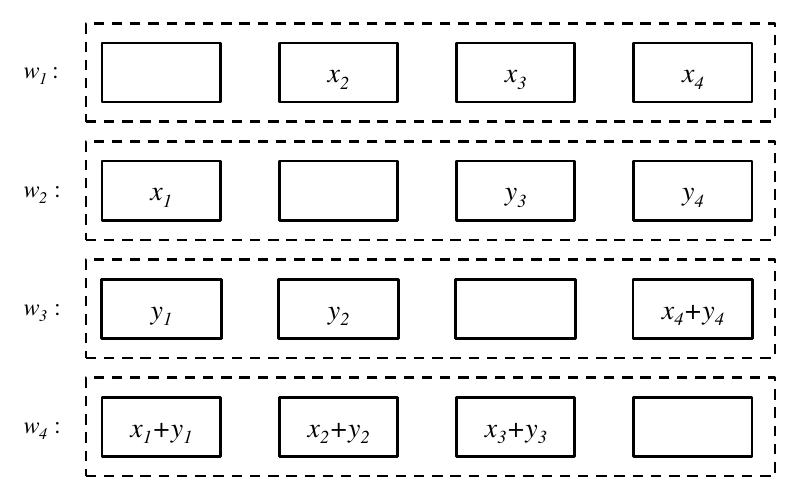}
     \caption{The figure illustrates the DSS built in Example \ref{helppoesim}. It consists of nodes $w_1$, $w_2$, $w_3$, and $w_4$.}

\end{figure}

\section{Inequalities from the Construction}\label{inequalities}
Next we will derive some inequalities for the capacity in the case of exact repair. Using Equation \ref{firstbound} inductively we get
\begin{thm}\label{bound}
For an integer $j \in [ 0,k-1 ]$ we have
$$
C^{\text{exact}}_{n,k,d}\left(\alpha,\gamma\right) \geq \frac{n}{n-j} C^{\text{exact}}_{n-j,k-j,d-j}(\alpha,\gamma).
$$
\end{thm}

It is proved in \cite{kumar}, \cite{MSRequal}, and \cite{optimalMDS} that the MSR point can be achieved if exact repair is assumed. As a consequence of this and Theorem \ref{firstbound} we get the following bound.
\begin{thm}\label{boundMSR}
For integers $1 \leq i \leq k$ we have
$$
C^{\text{exact}}_{n,k,d}\left(\alpha,\frac{(d-k+i)\alpha}{d-k+1}\right) \geq \frac{ni\alpha}{n-k+i}\,.
$$
\end{thm}
\begin{proof}
Write $n'=n-j,k'=k-j,d'=d-j$, $\alpha=\frac{B}{k'}$, and $\gamma=\frac{d' B}{k' (d'-k'+1)}$. It is proved in \cite{kumar}, \cite{MSRequal}, and \cite{optimalMDS} that
$$
C^{\text{exact}}_{n',k',d'}(\alpha,\gamma) = B,
$$
\emph{i.e.},
$$
C^{\text{exact}}_{n-j,k-j,d-j}\left(\alpha,\frac{(d-j) \alpha}{d-k+1}\right) = (k-j)\alpha.
$$

Hence by Theorem \ref{bound} we have
$$
C^{\text{exact}}_{n,k,d}\left(\alpha,\frac{(d-j) \alpha}{d-k+1}\right) \geq \frac{n(k-j)\alpha}{n-j}\,.
$$

Now a change of  variables by setting $i=k-j$ gives us  the  result.
\end{proof}

\section{Example: Case $n=k+1=d+1$}\label{example}
In this section we study the special case $n=k+1=d+1$ and compare it to the known capacity with the assumption of functional repair,
$$
C_{n-1,n-1}(\alpha,\gamma) = \sum_{j=0}^{n-2} \min \left\{ \alpha , \frac{n-1-j}{n-1} \gamma \right\}.
$$

Now our bound gives
$$
C^{\text{exact}}_{n,n-1,n-1}(\alpha,i\alpha) \geq \frac{ni\alpha}{1+i}
$$
so  we can write
$$
f_{n} (i) = \frac{ni\alpha}{1+i}
$$
for integers $i=1,\dots,k$. 

Notice that now in the extreme points our lower bound achieves the known capacity, \emph{i.e.},
$$
C^{\text{exact}}_{n,n-1,n-1}(\alpha,\alpha) = f_{n} (1) = \frac{n\alpha}{2}
$$
for the MBR point and
$$
C^{\text{exact}}_{n,n-1,n-1}(\alpha,k\alpha) = f_{n} (k) = (n-1)\alpha
$$
for the MSR point.

As an example we study the fraction
$$
\frac{f_{n} (i)}{C_{n-1,n-1}(\alpha,i\alpha)} = \frac{\frac{ni\alpha}{1+i}}{\sum_{j=0}^{n-2} \min \left\{ \alpha , \frac{n-1-j}{n-1} i\alpha \right\}}
$$
for integers $i \in [1,k]$. Writing it out we see that
\begin{equation}
\begin{split}
&\frac{f_{n} (i)}{C_{n-1,n-1}(\alpha,i\alpha)} \\
= & \frac{\frac{ni}{1+i}}{\sum_{j=0}^{T } 1 + \sum_{j=T + 1}^{n-2}  \frac{n-1-j}{n-1} i } \\
= & \frac{\frac{ni}{1+i}}{ T +1 + \frac{i}{2(n-1)} \cdot (n-T-1)(n-T-2)}\,, \\
\end{split}
\end{equation}
where $T=\lfloor (n-1)(1-\frac{1}{i}) \rfloor$

For large values of $n$ this is approximately
$$
\frac{2i^2}{2i^2+i-1} \geq \frac{8}{9}
$$
for all $i=1,\dots,k$.

\begin{figure}[ht]
     \includegraphics[width=9cm] {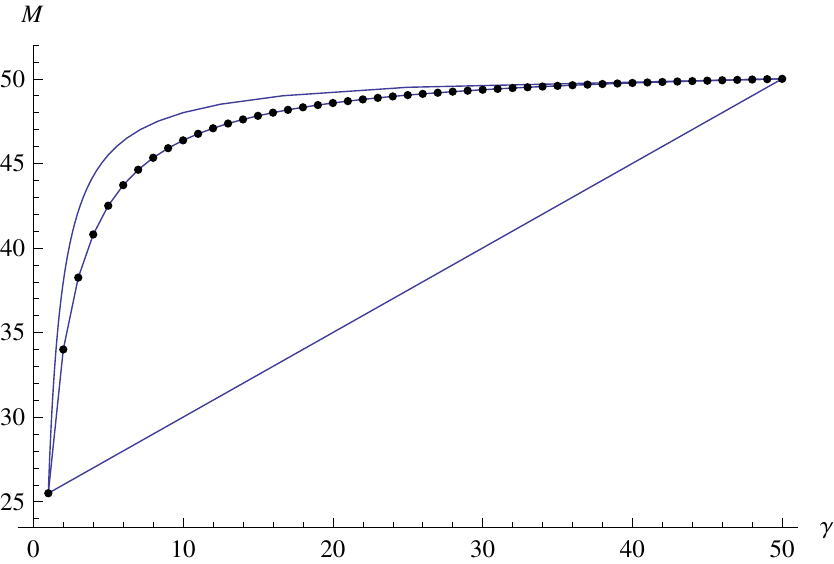}
     \caption{The figure shows the performance $M$ of our construction (dotted curve) between the capacity of functionally repairing codes (uppermost curve) and the trivial lower bound given by interpolation of the known MSR and MBR points when $(n,k,d)=(51,50,50)$, $\alpha=1$, and $\gamma \in [1,50]$.}

\end{figure}

\section{The case when $n$, $k$ and $d$ are close to each other}\label{analysis}
Next  we will study the special case where $n$, $k$ and $d$ are close to each other. We will do this by setting $n_M=n+M$, $k_M=k+M$ and $d_M=d+M$ and letting $M \rightarrow \infty$, and then examine how the capacity curve asymptotically behaves. The example in the previous section showed us that in that special case our bound is quite close to the capacity of functionally regenerating codes. However, in the previous section we fixed $i$ to be an integer and then assumed that $n$ is large. In this section we tie up the values $i$ and $M$ together to arrive at a situation where the total repair bandwidth stays on a fixed point between its minimal possible value given by the MBR point and its maximal possible value given by the MSR point.

For each $M$ the bound from Theorem \ref{boundMSR} gives
$$
C^{\text{exact}}_{n_M,k_M,d_M}\left(\alpha,\frac{(d_M-k_M+i)\alpha}{d_M-k_M+1}\right) \geq \frac{n_M i \alpha}{n-k+i}
$$
for $i=1,\dots,k_M$, hence in this section we write
$$
g_{M} (i) = \frac{n_M i \alpha}{n-k+i}
$$
for integers $i=1,\dots,k$ and extend this definition for $x \in [1,k]$ such that $g_{M} (x)$ is the piecewise linear curve defined by $g_{M} (i)$.

Let $s \in (0,1]$ be a fixed number and $i=1+s(k_M-1)$. We will study how the fraction
$$
\frac{g_{M} (i)}{C_{k_M,d_M}(\alpha,\frac{(d_M-k_M+i)\alpha}{d_M-k_M+1})}
$$
behaves as we let $M \rightarrow \infty$. Informally this tells how close our lower bound curve and the known capacity curve are to each other when $M$ is large, \emph{i.e.}, values $n_M,k_M,d_M$ are close to each other.

\begin{remark}
In the MSR point we have $$\gamma_{MSR}=\frac{d_M \alpha}{d_M-k_M+1}$$ and in the MBR point $$\gamma_{MBR}=\alpha.$$ Hence
$$
\alpha \cdot \frac{d_M-k_M+i}{d_M-k_M+1} = s\gamma_{MSR}+(1-s)\gamma_{MBR}.
$$
\end{remark}

\begin{thm}\label{asymptotic}
Let $s \in (0,1]$ be a fixed number and $i=1+s(k_M-1)$. Then
$$
\lim_{M \rightarrow \infty} \frac{g_{M} (i)}{C_{k_M,d_M}(\alpha,\frac{(d_M-k_M+i)\alpha}{d_M-k_M+1})} = 1.
$$
\end{thm}
\begin{proof}
Let $i=1+s(k_M-1)$. We study the behavior of the fraction for large $M$, so we have $\frac{\lfloor i \rfloor}{i} \approx 1$. Thus, to simplify the notation,  we may assume that $i$ acts as an integer. We also use the notation
$$
t=\frac{d_Ms(k_M-1)}{d-k+1+s(k_M-1)}.
$$

We have
$$
g_{M} (1+s(k_M-1))=\frac{n_M (1+s(k_M-1)) \alpha}{n-k+i}
$$
and
\begin{equation}
\begin{split}
& C_{k_M,d_M}(\alpha,\frac{(d_M-k_M+i)\alpha}{d_M-k_M+1}) \\
= & \alpha \left( \sum_{j=0}^{t} 1 + \sum_{j=t+1}^{k_M-1} \frac{d_M-j}{d_M} \cdot \frac{d-k+i}{d-k+1} \right) \\
= & \alpha  \left( t+1 +  \frac{(k_M-t-1)(2d+M-k-t)(d-k+i)}{2 d_M (d-k+1)} \right), \\
\end{split}
\end{equation}
whence
\begin{equation}
\begin{split}
& \frac{g_{M} (i)}{C_{k_M,d_M}(\alpha,\frac{(d_M-k_M+i)\alpha}{d_M-k_M+1})} \\
= & \frac{h_1(M)}{h_2(M)(h_3(M)+h_4(M))}, \\
\end{split}
\end{equation}
where
$$
h_1(M)=2n_M (1+s(k_M-1)) d_M (d-k+1),
$$
$$
h_2(M)=n-k+1+s(k_M-1),
$$
$$
h_3(M)=2(t+1)d_M(d-k+1),
$$
and
$$
h_4(M)=(k_M-t-1)(2d-k+M-t)(d-k+1+s(k_M-1)).
$$

Now it is easy to check that
$$
\frac{h_1(M)}{M^3} \rightarrow 2s(d-k+1),
$$
$$
\frac{h_2(M)}{M} \rightarrow s,
$$
and
$$
\frac{h_3(M)}{M^2} \rightarrow 2(d-k+1)
$$
as $M \rightarrow \infty.$

Note that $$M-t \approx \frac{d-k+1-ds}{s}$$ when $M$ is large and hence
\begin{equation}
\begin{split}
& \frac{h_4(M)}{M^2} \\
= & \frac{(k_M-t-1)(2d-k+M-t)}{M} \cdot \frac{d-k+1+s(k_M-1)}{M} \\
\rightarrow & 0 \cdot s=0 \\
\end{split}
\end{equation}
as $M \rightarrow \infty.$

Finally,
\begin{equation}
\begin{split}
& \frac{g_{M} (i)}{C_{k_M,d_M}(\alpha,\frac{(d_M-k_M+i)\alpha}{d_M-k_M+1})} \\
= & \frac{\frac{h_1(M)}{M^3}}{\frac{h_2(M)}{M} \cdot \frac{h_3(M)+h_4(M)}{M^2}} \\
\rightarrow & \frac{2s(d-k+1)}{s(2(d-k+1)+0)}=1 \\
\end{split}
\end{equation}
as $M \rightarrow \infty$, proving the claim.

\end{proof}

As a straightforward corollary to Theorem \ref{asymptotic} we have

\begin{thm}\label{asymptotic}
Let $s \in [0,1]$ be a fixed number and let $\gamma_{MSR}=\frac{d_M \alpha}{d_M-k_M+1}$ and $\gamma_{MBR}=\alpha$. Then
$$
\lim_{M \rightarrow \infty} \frac{C^{\text{exact}}_{n_M,k_M,d_M}(\alpha,s\gamma_{MSR}+(1-s)\gamma_{MBR})}{C_{k_M,d_M}(\alpha,s\gamma_{MSR}+(1-s)\gamma_{MBR})} = 1.
$$
\end{thm}

\section{Conclusions}
We have shown in this paper that when $n$, $k$, and $d$ are close to each other, the capacity of a distributed storage system when exact repair is assumed is essentially the same as when only functional repair is required. This was proved by using a specific code construction exploiting some already known codes achieving the MSR point on the tradeoff curve and by studying the asymptotic behavior of the capacity curve.

However, when $n$, $k$, and $d$ are not close to each other then the bound our construction gives is not good. So as a future work it is still left to find the precise expression of the capacity of a distributed storage system when exact repair is assumed, and especially to study the behavior of the capacity when $n$, $k$, and $d$ are not close to each other.

\section{Acknowledgments} This research was partly supported by the Academy of Finland (grant \#131745) and by the Emil Aaltonen Foundation, Finland, through grants to Camilla Hollanti.

Dr. Salim El Rouayheb at the Princeton University is gratefully acknowledged for useful discussions. Dr. Camilla Hollanti at the Aalto University is gratefully acknowledged for useful comments on the first draft of this paper.


\end{document}